\documentclass[twocolumn,pra,aps,superscriptaddress,showpacs]{revtex4}
\usepackage[T1]{fontenc}
\usepackage{amsmath}
\usepackage{amssymb}
\usepackage{graphicx}
\usepackage{array}
\usepackage{amsthm}
\newtheorem{theorem}{Theorem}

\newtheorem{lemma}[theorem]{Lemma}

\begin{document}

\title{Necessary and sufficient condition for saturating the upper bound
of quantum discord}

\author{Zhengjun Xi}
\affiliation{College of Computer Science, Shaanxi Normal University, Xi'an, 710062,
China}

\author{Xiao-Ming Lu}
\affiliation{Centre for Quantum Technologies, National University of Singapore,
3 Science Drive 2, Singapore 117543}

\author{Xiaoguang Wang}
\affiliation{Zhejiang Institute of Modern Physics, Department of Physics, Zhejiang
University, Hangzhou, 310027, China}

\author{Yongming Li}
\affiliation{College of Computer Science, Shaanxi Normal University, Xi'an, 710062,
China}

\begin{abstract}
We revisit the upper bound of quantum discord given by the von Neumann
entropy of the measured subsystem. Using the Koashi-Winter relation,
we obtain a trade-off between the amount of classical correlation
and quantum discord in the tripartite pure states. The difference
between the quantum discord and its upper bound is interpreted as
a measure on the classical correlative capacity. Further, we give
the explicit characterization of the quantum states saturating the
upper bound of quantum discord, through the equality condition for
the Araki-Lieb inequality. We also demonstrate that the saturating of the
upper bound of quantum discord precludes any further correlation between
the measured subsystem and the environment.
\end{abstract}


\maketitle

\section{Introduction}

Recently, it has been recognized that entanglement does not depict
all possible quantum correlations contained in a bipartite state.
Based on the measurement on the subsystem in the bipartite system, quantum
discord was proposed as a measure for the quantum correlation beyond
entanglement~\cite{Zurek00,Ollivier02,Vedral02,Celeri11,Modi2011}. Quantum
discord was viewed as a figure of merit for characterizing the nonclassical
resources in the deterministic quantum computation with one-qubit~\cite{Knill98,Datta08}.
It was also discussed that zero-discord of the initial system-environment
states is a necessary and sufficient condition for completely positivity
of reduced dynamical maps~\cite{Rosario08,Shabani09}. At the same
time, a necessary and sufficient condition for nonzero-discord was
also given for any dimensional bipartite states~\cite{Vedral11}.
Most recently, operational interpretations of quantum discord were
proposed in~\cite{Datta11,Winter11}, where quantum discord was shown
to be a quantitative measure about the performance in the quantum
state merging~\cite{Horodecki05}. Over the past decade, quantum
discord has received a lot of attentions in Refs.~\cite{Cen11,Modi10,Lu11,Luo10,Adesso11,Guzi11,Fei11,Chen11,
RioNature,Ferraro10,Ali10,Streltsov11,Piani11,Cornelio11,Fanchini11a,Mazzola10,Lanyon08},
and also see the reviews \cite{Modi2011,Celeri11}.

In general, quantum discord is upper bounded by the entropy of the
measured subsystem~\cite{Datta1003,Xi11,LuoPRAa,LuoPRAb,Terno11},
but it has remained an open question when this bound is saturated
for general mixed states. For this question, some sufficient conditions
were given in~\cite{Xi11,LuoPRAa,LuoPRAb,Zhang11}. This motivates
us to systematically investigate the upper bound of quantum discord
and give a necessary and sufficient condition for saturating the upper
bound of quantum discord. On the other hand, as an achievable upper
bound of the quantum discord, the von Neumann entropy of a system
can be considered as the quantum correlative capacity which tells
us how strongly can this system be quantum correlated with others~\cite{LuoPRAa}.
So we want to ask another question: what can the closeness between the quantum
discord and the upper bound tell us?

It is interesting that there are monogamic relations between different
measures on correlations~\cite{Koashi04}. In this paper, by purifying
the bipartite systems and using the Koashi-Winter relation~\cite{Koashi04},
we obtain a new monogamic relation: a system being quantum correlated
with another one limits its possible classical correlation with a
third system. For the bipartite systems, the total amount of quantum
discord between two subsystems and the classical correlation between
the measured subsystem and the environment cannot exceed the von Neumann
entropy of the measured subsystem. We further prove that the necessary and sufficient condition for saturating the upper bound of quantum
discord is equivalent to the equality condition for the Araki-Lieb
inequality~\cite{AL70,Nielsen,LR73a,LR73b}. And we give the explicit
characterization of the quantum states saturating the upper bound
of quantum discord, through which we demonstrate that saturating the
upper bound of quantum discord means that the measured subsystem cannot
be further correlated with the environment. For a two-qubit system,
quantum discord is strictly less than the von Neumann entropy
of the measured qubit of two-qubit states other than when the two-qubit
system in pure states.

This paper is organized as follows. In Sec. \ref{sec:review QD},
we give a brief review on quantum discord. In Sec. \ref{sec:UP_QD},
we show that there is a trade-off between the quantum discord and
the classical correlation, and then we prove a necessary and sufficient condition for the saturating of the upper bound of quantum
discord. Section \ref{sec:conclusion} is the conclusion.

\section{Review of quantum discord}

\label{sec:review QD}

First, we recall the concepts and properties of the quantum discord.
Let Hilbert space $\mathcal{H}=\mathcal{H}_{A}\otimes\mathcal{H}_{B}$
be a bipartite quantum system. Let $\mathcal{D}(\mathcal{H})$ be
a set of bounded, positive-semidefinite operators with unit trace
on $\mathcal{H}$. Given a bipartite quantum state $\rho^{AB}\in\mathcal{D}(\mathcal{H})$,
the von Neumann mutual information between the two subsystems $A$
and $B$ is defined as~\cite{Nielsen}
\begin{equation}
I(A:B):=S(A)+S(B)-S(AB),\label{eq:TC}
\end{equation}
where $S(X):=-\mathrm{Tr}\rho^{X}\log_{2}\rho^{X}$ is
the von Neumann entropy,  $\rho^{X}$ is a quantum state of system $X$~\cite{Nielsen}. The mutual information $I(A:B)$
is a measure of the total amount of correlations in the bipartite
quantum state. Generally speaking, it was divided into quantum correlation
and classical correlation~\cite{Vedral02,Vedral03,Zurek00,Ollivier02}.

The classical correlation was seen as the amount of information about
the subsystem $A$ that can be obtained via performing a measurement
on the other subsystem $B$. Then, the measure of the classical correlation~\cite{Vedral02}
is defined by
\begin{equation}
J(A|B):=\max_{\{E_{i}^{B}\}}\Big(S(A)-\sum_{i}p_{i}S(A|E_{i}^{B})\Big),\label{eq:CC}
\end{equation}
 where $\{E_{i}^{B}\}$ is the positive operator valued measure (POVM)
on $B$, $S(A|E_{i}^{B})$ is the von Neumann entropy of the post-measurement
states $\rho_{i}^{A}=$Tr$_{B}(E_{i}^{B}\rho^{AB})$ corresponds to
outcome $i$ with the probability $p_{i}=$Tr$(E_{i}^{B}\rho^{AB})$.
Therefore, quantum discord~\cite{Zurek00,Ollivier02} is defined by the
difference of the quantum mutual information and the classical correlation
\begin{equation}
D(A|B):=I(A:B)-J(A|B).\label{eq:QDA_B}
\end{equation}

Quantum discord is asymmetric with respect to $A$ and $B$,
in general, $D(A|B)\neq D(B|A)$, and it is always nonnegative~\cite{Zurek00,Ollivier02}. Quantum discord
vanishes if and only if there exists an optimal choice of measurement
$\{E_{i}^{B}\}$ on $B$ leaves the state $\rho^{AB}$ unperturbed~\cite{Terno11}.
The condition for zero-discord can been reduced to the equality condition
using relative entropy~\cite{Hayden04,Datta1003}, and was also discussed
in Refs~\cite{Ollivier02,Vedral11,Datta1003}. For any bipartite
pure state $|\psi\rangle^{AB}$, one checks that
\begin{equation}
D(B|A)=S(A)=S(B)=D(A|B).
\end{equation}

\section{The upper bound of quantum discord}
\label{sec:UP_QD}
\subsection{A trade-off between quantum discord and classical correlation}

For any general bipartite mixed state $\rho^{AB}$, we can always
find a tripartite pure state $\rho^{ABE}=|\Psi\rangle^{ABE}\langle\Psi|$
such that $\rho^{AB}=\mathrm{Tr}_{E}(\rho^{ABE})$, where $E$ represents
the environment. Hereafter, we will consider this purification about
general bipartite mixed state $\rho^{AB}$. The monogamic relation
between the entanglement of formation and the classical correlation
between the two subsystems is given by~\cite{Koashi04} \begin{subequations}
\begin{align}
E_{F}(A:E)+J(A|B)=S(A),\label{eq:KW2}
\end{align}
where $E_{F}(A:E)$ is entanglement of formation (EoF), defined as
$E_{F}(A:E)=\min_{\{p_{i},|\psi_{i}\rangle\}}\sum_{i}p_{i}S(\mathrm{Tr}_{E}(|\psi_{i}\rangle\langle\psi_{i}|))$,
where the minimum is taken over all pure ensembles $\{p_{i},|\psi_{i}\rangle\}$
satisfying $\rho^{AE}=\sum_{i}p_{i}|\psi_{i}\rangle\langle\psi_{i}|$~\cite{Bennett96,Wootters98,Wootters01}.

This relation is universal for any tripartite pure states. Thus, we
can give an other reorder version of Koashi-Winter relation,
\begin{align}
E_{F}(E:A)+J(E|B)=S(E).\label{eq:KW6}
\end{align}
 \end{subequations} Due to the symmetric property of the EoF, i.e.,
$E_{F}(A:E)=E_{F}(E:A)$, we can eliminate the EoF by combining Eqs.~(\ref{eq:KW2})
and (\ref{eq:KW6}), then we obtain
\begin{align}
J(A|B)-J(E|B)=S(A)-S(E).
\end{align}
 To be clearer, we substitute $D(A|B)=I(A:B)-J(A|B)$ into the above
equation, and obtain a trade-off between the quantum discord and the classical correlation
as follows
\begin{equation}
D(A|B)+J(E|B)=S(B).\label{eq:KWX_1}
\end{equation}
 This new monogamic relation tell us that the amount of quantum correlation
between $A$ and $B$, plus the amount of classical correlation between
$B$ and the complementary part $E$, must be equal to the entropy
of the measured subsystem $B$. More importantly, based on this new
monogamic equation, we can introduce
\begin{equation}
\tilde{J}(B/A):=S(B)-D(A|B),
\end{equation}
which quantifies classically correlative capacity of $B$ with other
systems except $A$. In other words, for general tripartite mixed
states $\rho^{ABC}$, the classical correlation between $B$ and $C$
cannot be greater than $\tilde{J}(B/A)$. To be convinced, let us
purify $\rho^{ABC}$ as $\rho^{ABC}=\mathrm{Tr}_{E}|\Psi\rangle^{ABCE}\langle\Psi|$,
then we have $J(CE|B)=\tilde{J}(B/A)$ due to the monogamic relation
(\ref{eq:KWX_1}). Because the classical correlation is non-increasing
under the local quantum operation~\cite{Vedral02}, then we have
\begin{equation}
J(C|B)\leq J(CE|B)=\tilde{J}(B/A),
\end{equation}
which is equivalent to
\begin{equation}
D(A|B)+J(C|B)\le S(B).\label{eq:tradeoff_mixed_states}
\end{equation}
We emphasize that Eq.~(\ref{eq:tradeoff_mixed_states}) is applicable
for any tripartite state $\rho^{ABC}$ and the equality holds if $\rho^{ABC}$
is pure. We can see that the total amount of the quantum discord between
two subsystems and the classical correlation between the measured
subsystem and the environment cannot exceed the von Neumann entropy
of the measured subsystem.

\subsection{The upper bound of quantum discord}

The monogamic relation (\ref{eq:KWX_1}) directly supplies a general
upper bound for the quantum discord, which was proved in Refs.~\cite{Datta1003,Xi11}.
In the following, we are going to determine which states saturate
this bound. Combining the monogamic relation (\ref{eq:KWX_1}) with
the equality condition for Araki-Lieb inequality, we have the following
result.

\begin{theorem}
\label{theorem:1} For the bipartite state $\rho^{AB}$, we have
\begin{equation}
D(A|B)\leq S(B)\label{eq:qd_leq_SB}
\end{equation}
 with equality if and only if there exist a decomposition of $\mathcal{H}^{A}$
as $\mathcal{H}^{A^{L}}\otimes\mathcal{H}^{A^{R}}$ such that
\begin{equation}
\rho^{AB}=\rho^{A^{L}}\otimes|\psi\rangle^{A^{R}B}\langle\psi|.
\end{equation}
\end{theorem}

To prove this theorem, we first introduce a lemma as follows.

\begin{lemma} \label{lemma:1}Considering different correlations
and entropy in the bipartite states, the following conditions are
equivalent
\begin{enumerate}
\item $D(A|B)=S(B)$;\label{enu:equivalence1}
\item $S(A)-S(B)=S(AB)$;\label{enu:equivalence2}
\item $E_{F}(A:B)=S(B)$.\label{enu:equivalence3}
\end{enumerate}
If one of them is satisfied, then the others are satisfied.
\end{lemma}
\begin{proof}
The equivalence of these conditions can be proved by the trade-off
relation (\ref{eq:KWX_1}) for the purification $|\psi^{ABE}\rangle$.
From Eq. (\ref{eq:KWX_1}), we can see that $D(A|B)=S(B)$ is equivalent
to $J(E|B)=0$. It is known that the classical correlation vanishes
if and only is the states are product states, see Refs.~\cite{Vedral02,Vedral03}.
So we have
\begin{eqnarray}
J(E|B)=0 & \Leftrightarrow & \rho^{EB}=\rho^{E}\otimes\rho^{B}\label{eq:equavalence_EB_product}\\
 & \Leftrightarrow & S(EB)=S(B)+S(E).\label{eq:equivalence2_0}
\end{eqnarray}
 For the tripartite pure states, we have $S(A)=S(EB)$ and $S(E)=S(AB)$.
Thus Eq.(\ref{eq:equivalence2_0}) is equivalent to $S(A)-S(B)=S(AB)$
and we obtain the equivalence between the first and second conditions.
Meantime, the third condition is equivalent to $J(B|E)=0$ due to
the Koashi-Winter relation $E_{F}(A:B)=S(B)-J(B|E)$. Again using
the property of the classical correlation, we can see the third condition
is equivalent to $\rho^{EB}=\rho^{E}\otimes\rho^{B}$, and in turn
equivalent to the first condition.
\end{proof}
Form the above lemma, we will complete the proof of Theorem~\ref{theorem:1}.
\begin{proof}[Proof of Theorem~\ref{theorem:1}]
The inequality was recently given in~\cite{Xi11,Datta1003}, and
it is also a direct consequence of Eq.(\ref{eq:KWX_1}).
Therefore, we obtain that quantum discord $D(A|B)$ saturates the upper bound $S(B)$
if and only if $S(A)-S(B)=S(AB)$, which is the equality condition
of the Araki-Lieb inequality \cite{AL70}
\begin{equation}
\left|S(A)-S(B)\right|\leq S(AB)
\end{equation}
 when $S(A)\geq S(B)$.

 Recently, the quantum states saturated the
Araki-Lieb inequality was explicitly given by Zhang and Wu \cite{Zhang11},
by using the explicit characterization of quantum states that saturates
the strong subadditivity inequality \cite{Hayden04} and the relation
between the Araki-Lieb inequality and the strong subadditivity inequality
\cite{Nielsen}. From the result in Ref~\cite{Zhang11}, we know
that $S(A)-S(B)=S(AB)$ if and only if there exist a decomposition
$\mathcal{H}^{A}=\mathcal{H}^{A^{L}}\otimes\mathcal{H}^{A^{R}}$ such
that
\begin{equation}
\rho^{AB}=\rho^{A^{L}}\otimes|\varphi\rangle^{A^{R}B}\langle\varphi|.
\end{equation}
This completes this proof of Theorem~\ref{theorem:1}.
 \end{proof}

\begin{figure}
\centering{}\includegraphics[scale=0.8]{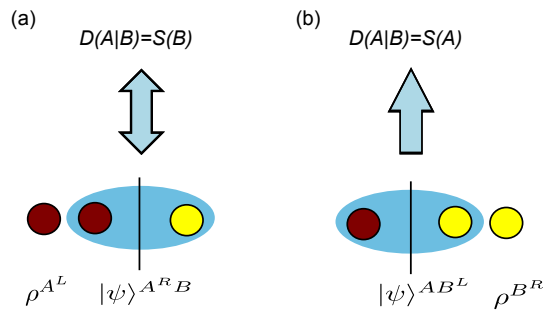} \caption{Schematics diagram for the Theorem~\ref{theorem:1} and~\ref{theorem:2}.
In both subfigures, the right side of the solid line are the measured
subsystems B. In (a), the quantum discord saturates the general upper
bound given by the von Neumann entropy of the measured subsystem $B$
if and only if there exists a decomposition of $\mathcal{H}^{A}$
as $\mathcal{H}^{A^{L}}\otimes\mathcal{H}^{A^{R}}$ such that $\rho^{AB}=\rho^{A^{L}}\otimes|\varphi\rangle^{A^{R}B}\langle\varphi|$;
In (b), the quantum discord equals the von Neumann entropy of the
unmeasured subsystem $A$ if there exist a decomposition of $\mathcal{H}^{B}$
as $\mathcal{H}^{B^{L}}\otimes\mathcal{H}^{B^{R}}$ such that $\rho^{AB}=|\psi\rangle^{AB^{L}}\langle\psi|\otimes\rho^{B^{R}}$.
\label{fig:trap}}
\end{figure}

This result shows that quantum discord is equal to the measured system
entropy if and only if the equality in the Araki-Lieb inequality holds.
It is an interesting thing that the measured subsystem $B$ cannot
correlate with the environment $E$ if the quantum discord between
$A$ and $B$ is equal to the entropy of the measured subsystem. In
other words, there must exists an isolated pure subsystem enclosing
the measured subsystem in the Hilbert space $\mathcal{H}_{A}\otimes\mathcal{H}_{B}$
when $D(A|B)=S(B)$. One can also see that given the entropy of the
measured subsystem, the maximal quantum discord between $A$ and $B$
will forbid system $B$ from being correlated to other systems outside
this composite system.

Due to the asymmetric property of quantum discord, one may ask whether
or not Theorem~\ref{theorem:1} still holds if we consider
$S(A)$ instead of $S(B)$. In fact, a conjecture about the von Neumann
entropy and quantum discord was presented by Luo, Fu and Li in Ref.~\cite{Luo10},
namely,
\begin{equation}
D(A|B)\leq\min\{S(A),S(B)\}.
\end{equation}
Later, Li and Luo showed that the part $D(A|B)\leq S(A)$ fails in
general, but might be true for low-dimension systems~\cite{LuoPRAa}.

We only have the following sufficient condition for the situation of $D(A|B)=S(A)$.

\begin{theorem}\label{theorem:2} For the bipartite state $\rho^{AB}$,
we have
\begin{equation}
D(A|B)=S(A),
\end{equation}
 if the equality $S(B)-S(A)=S(AB)$ is satisfied. \end{theorem}

\begin{proof} According to the equality condition for the Araki-Lieb
inequality, $S(B)-S(A)=S(AB)$ is equivalent to that there exists a decomposition
$\mathcal{H}^{B}$ as $\mathcal{H}^{B^{L}}\otimes\mathcal{H}^{B^{R}}$
such that
\begin{equation}
\rho^{AB}=|\psi\rangle^{AB^{L}}\langle\psi|\otimes\rho^{B^{R}}.\label{eq:decomposition_B}
\end{equation}
For this density matrix, we can get
\begin{equation}
D(A|B)=D(A|B^{L})=S(A),
\end{equation}
where $D(A|B^{L})$ is the quantum discord for the pure state $|\psi\rangle^{AB^{L}}$.
\end{proof}

Besides, applying the Theorem \ref{theorem:1} and Lemma \ref{lemma:1},
we can see that all the following quantities are equal:
\begin{align}
D(A|B)=D(B|A)=E_{F}(A:B)=S(A)=S(B^{L})\label{eq:all_equals}
\end{align}
for the quantum state (\ref{eq:decomposition_B}).

As an illustration, let us consider the following example
\begin{equation}
\rho^{AB}=\frac{1}{4}\left(\begin{array}{cccccccc}
1 & 0 & 0 & 0 & 0 & 0 & 1 & 0\\
0 & 1 & 0 & 0 & 0 & 0 & 0 & 1\\
0 & 0 & 0 & 0 & 0 & 0 & 0 & 0\\
0 & 0 & 0 & 0 & 0 & 0 & 0 & 0\\
0 & 0 & 0 & 0 & 0 & 0 & 0 & 0\\
0 & 0 & 0 & 0 & 0 & 0 & 0 & 0\\
1 & 0 & 0 & 0 & 0 & 0 & 1 & 0\\
0 & 1 & 0 & 0 & 0 & 0 & 0 & 1
\end{array}\right),
\end{equation}
where $\mathcal{H}_{A}=\mathbb{C}^{2}$ and $\mathcal{H}_{B}=\mathbb{C}^{4}$.
The reduced states can be obtained $\rho^{A}=\frac{I^{A}}{2}$ and
$\rho^{B}=\frac{I^{B}}{4}$, where $I^{A}$ and $I^{B}$ are identity
operators on $\mathcal{H}_{A}$ and $\mathcal{H}_{B}$, respectively.
After some calculations one obtains
\begin{equation}
S(AB)=1,\: S(A)=1,\: S(B)=2.
\end{equation}
Therefore, this state satisfies $S(B)-S(A)=S(AB)$. On the other hand,
one checks that $\mathcal{H}_{B}=\mathcal{H}_{B^{L}}\otimes\mathcal{H}_{B^{R}}$
with $\mathcal{H}_{B^{L}}=\mathcal{H}_{B^{R}}=\mathbb{C}^{2}$. Then,
$\rho^{AB}$ is in the form of Eq.~(\ref{eq:decomposition_B}) with
$|\psi\rangle^{AB^{L}}=\frac{1}{\sqrt{2}}\big(|00\rangle+|11\rangle\big)$
and $\rho^{B^{R}}=\frac{I^{B^{R}}}{2}$, where $I^{B^{R}}$ is identity
operator on $\mathcal{H}_{B^{R}}$. With the result of Eq.~(\ref{eq:all_equals}),
we have
\begin{equation}
D(A|B)=D(B|A)=E_{F}(A:B)=1
\end{equation}
for this quantum state.

Through Theorem~\ref{theorem:1} and Theorem~\ref{theorem:2}, we identify
two conditions to witness the trapping of correlation within a pure-state
subsystem, see Fig. \ref{fig:trap}. This opens a new way to investigate
whether the correlation between A and B is essentially the correlation
between smaller parts of them.

\subsection{Arbitrary two-qubit states}

Now, we will discuss the application of our results in two-qubit system.
A two-dimension Hilbert space cannot be decomposed any more. The only
possibility of $|S(A)-S(B)|=S(AB)$ for two-qubit states is that $\rho^{AB}$
is a pure state. Hence, its upper bound is not reachable except two-qubit
pure state~\cite{Xi11}, namely, we have
\begin{equation}
D(B|A)<S(A),D(A|B)<S(B),
\end{equation}
for any mixed two-qubit state.

\section{conclusion}

\label{sec:conclusion}

In this work, we have given a new monogamic relation between  the quantum discord and the classical
correlation in terms of the Koashi-Winter relations.
Based on the equality conditions for the Araki-Lieb inequality, we have
given a necessary and sufficient condition for the saturating of the
upper bound of quantum discord. We have shown that the subsystem of
the bipartite system can not correlated with the other system if the
quantum discord of the bipartite system is equal to the entropy of
the measured subsystem. We showed that there are some mixed states
where quantum discord is equal to entanglement of formation. In particular,
for two-qubit state, its upper bound is not reachable except pure
state.

\section{acknowledgments}

The authors are very grateful to the referees for helpful comments and criticisms. We thank A. Datta for interesting discussions. Z. J. Xi is supported
by the Superior Dissertation Foundation of Shaanxi Normal University
(S2009YB03). X.-M. Lu is supported by National Research Foundation and
Ministry of Education, Singapore (Grant No. WBS: R-710-000-008-271). X. Wang is supported by NSFC with
Grants No.11025527, No.10874151, and No.10935010.
Y. M. Li is supported by NSFC with Grant No.60873119, and the Higher
School Doctoral Subject Foundation of Ministry of Education of China
with Grant No.200807180005.

\end{document}